\documentclass[10pt]{article}

\usepackage{times}
\usepackage{wrapfig}
\usepackage{citesort}
\usepackage{graphicx}
\usepackage{color}
\usepackage{fullpage}
\usepackage{amsmath}
\usepackage{amssymb}

\usepackage{theorem}

\newcommand{\eps}{\epsilon}

\allowdisplaybreaks

\newtheorem{theorem}{Theorem}
\newtheorem{lemma}{Lemma}

{\theorembodyfont{\rmfamily}
}
{\theorembodyfont{\rmfamily}
\newtheorem{definition}{Definition}}

\newenvironment{proof}{\trivlist\item[]\emph{Proof}:}%
{\unskip\nobreak\hskip 1em plus 1fil\nobreak$\Box$
\parfillskip=0pt%
\endtrivlist}

\newcommand{\A}{\mathcal{A}}
\newcommand{\M}{\mathcal{M}}
\newcommand{\Q}{\mathcal{Q}}
\newcommand{\B}{\mathcal{B}}

\renewcommand{\S}{\mathcal{S}}

\newcommand{\nm}{\frac{N}{M}}

\renewcommand{\paragraph}[1]{\vspace{2mm}\noindent{\bf #1}}

\title{\bf Dynamic Indexability: The Query-Update Tradeoff for
  One-Dimensional Range Queries}

\author{Ke Yi}

\date{Department of Computer Science and Engineering \\
Hong Kong University of Science and Technology \\
Hong Kong, China 
}

\begin{document}

\maketitle

\begin{abstract}
  The {\em B-tree} is a fundamental secondary index structure that is
  widely used for answering one-dimensional range reporting queries.  Given
  a set of $N$ keys, a range query can be answered in $O(\log_B \nm +
  \frac{K}{B})$ I/Os, where $B$ is the disk block size, $K$ the output
  size, and $M$ the size of the main memory buffer.  When keys are inserted
  or deleted, the B-tree is updated in $O(\log_B N)$ I/Os, if we require
  the resulting changes to be committed to disk right away.  Otherwise, the
  memory buffer can be used to buffer the recent updates, and changes can
  be written to disk in batches, which significantly lowers the amortized
  update cost.  A systematic way of batching up updates is to use the
  {\em logarithmic method}, combined with {\em fractional cascading},
  resulting in a dynamic B-tree that supports insertions in
  $O(\frac{1}{B}\log\nm)$ I/Os and queries in $O(\log\nm + \frac{K}{B})$
  I/Os.  Such bounds have also been matched by several known dynamic B-tree
  variants in the database literature.  Note that, however, the query cost
  of these dynamic B-trees is substantially worse than the $O(\log_B\nm +
  \frac{K}{B})$ bound of the static B-tree by a factor of $\Theta(\log B)$.
  
  In this paper, we prove that for any dynamic one-dimensional range query
  index structure with query cost $O(q+\frac{K}{B})$ and amortized
  insertion cost $O(u/B)$, the tradeoff $q\cdot \log(u/q) = \Omega(\log B)$
  must hold if $q=O(\log B)$.  For most reasonable values of the
  parameters, we have $\nm = B^{O(1)}$, in which case our query-insertion
  tradeoff implies that the bounds mentioned above are already optimal.  We
  also prove a lower bound of $u \cdot \log q = \Omega(\log B)$, which is
  relevant for larger values of $q$.  Our lower bounds hold in a dynamic
  version of the {\em indexability model}, which is of independent
  interests.  Dynamic indexability is a clean yet powerful model for
  studying dynamic indexing problems, and can potentially lead to more
  interesting complexity results.
\end{abstract}


\section{Introduction}

The {\em B-tree} \cite{bayer:organization} is a fundamental secondary index
structure used in nearly all database systems.  It has both very good space
utilization and query performance: Assuming each disk block can store $B$
data records, the B-tree occupies $O(\frac{N}{B})$ disk blocks for $N$ data
records, and supports one-dimensional range reporting queries in $O(\log_B
N + \frac{K}{B})$ I/Os (or page accesses) where $K$ is the output size.
Due to the large fanout of the B-tree, for most practical values of $N$ and
$B$, the B-tree is very shallow and $\log_B N$ is essentially a constant.
Very often we also have a memory buffer of size $M$, which can be used to
store the top $\Theta(\log_B M)$ levels of the B-tree, further lowering the
effective height of the B-tree to $O(\log_B \frac{N}{M})$, meaning that we
can usually get to the desired leaf with merely one or two I/Os, and then
start pulling out results.

If one wants to update the B-tree directly on disk, it is also well known
that it takes $O(\log_B N)$ I/Os.  Things become much more interesting if
we make use of the main memory buffer to collect a number of updates and
then perform the updates in batches, lowering the amortized update cost
significantly.  For now let us focus on insertions only; deletions are in
general much less frequent than insertions, and there are some generic
methods for dealing with deletions by converting them into insertions of
``delete signals'' \cite{arge:buffer,o'neil96:_lsm}.  The idea of using a
buffer space to batch up insertions has been well exploited in the
literature, especially for the purpose of managing historical data, where
there are much more insertions than queries.  The {\em LSM-tree}
\cite{o'neil96:_lsm} was the first along this line of research, by applying
the {\em logarithmic method} \cite{bs-dsp1s-80} to the B-tree.  Fix a
parameter $2\le \ell \le B$.  It builds a collection of B-trees of sizes up
to $M, \ell M, \ell ^2M, \dots$, respectively, where the first one always
resides in memory.  An insertion always goes to the memory-resident tree;
if the first $i$ trees are full, they are merged together with the
$(i+1)$-th tree by rebuilding.  Standard analysis shows that the amortized
insertion cost is $O(\frac{\ell}{B} \log_\ell \frac{N}{M})$.  A query takes
$O(\log_B N \log_\ell\frac{N}{M} + \frac{K}{B})$ I/Os since
$O(\log_\ell\frac{N}{M})$ trees need to be queried.  Using {\em fractional
  cascading} \cite{chazelle:fractionalI}, the query cost can be improved to
$O(\log_\ell\frac{N}{M} + \frac{K}{B})$ without affecting the (asymptotic)
size of the index and the update cost, but this result appears to be
folklore.  Later Jermaine et al. \cite{jermaine99:_novel} proposed the {\em
  Y-tree} as ``yet'' another B-tree structure for the purpose of lowering
the insertion cost.  The Y-tree is an $\ell$-ary tree, where each internal
node is associated with a bucket storing all the elements to be pushed down
to its subtree.  The bucket is emptied only when it has accumulated
$\Omega(B)$ elements.  Although \cite{jermaine99:_novel} did not give a
rigorous analysis, it is not difficult to derive that its insertion cost is
$O(\frac{\ell}{B} \log_\ell \frac{N}{M})$ and query cost
$O(\log_\ell\frac{N}{M} + \frac{K}{B})$, namely, the same as those of the
LSM-tree with fractional cascading.  Around the same time Buchsbaum et
al.~\cite{buschbaum:dfsdir} independently proposed the {\em buffered
  repository tree} in a different context, with similar ideas
and the same bounds as the Y-tree.  In order to support even faster
insertions, Jagadish et al.~\cite{jagadish97:_increm} proposed the {\em
  stepped merge tree}, a variant of the LSM-tree.  At each level, instead
of keeping one tree of size $\ell^i M$, they keep up to $\ell$ individual
trees.  When there are $\ell$ level-$i$ trees, they are merged to form a
level-$(i+1)$ tree.  The stepped merge tree has an insertion cost of
$O(\frac{1}{B}\log_\ell\frac{N}{M})$, lower than that of the LSM-tree.  But
the query cost is a lot worse, reaching $O(\ell\log_B N
\log_\ell\frac{N}{M} + \frac{K}{B})$ I/Os since $\ell$ trees need to be
queried at each level.  Again the query cost can be improved to $O(\ell
\log_\ell\frac{N}{M} + \frac{K}{B})$ using fractional cascading.  The
current best known results are summarized in Table~\ref{tab:results}.
Typically $\ell$ is set to be a constant
\cite{o'neil96:_lsm,jermaine99:_novel,jagadish97:_increm}, at which point
all the indexes have the same asymptotic performance of
$O(\log\frac{N}{M}+\frac{K}{B})$ query and $O(\frac{1}{B}\log\frac{N}{M})$
insertion.  Note that the amortized insertion bound of these dynamic B-trees
could be much smaller than one I/O, hence much faster than updating the
B-tree directly on disk.  The query cost is, however, substantially worse
than the $O(\log_B \nm)$ query cost of the static B-tree by an $\Theta(\log
B)$ factor.  As typical values of $B$ range from hundreds to thousands, we
are expecting a 10-fold degradation in query performance for these dynamic
B-trees.  Thus the obvious question is, can we lower the query cost while
still allowing for fast insertions?

\begin{table}[h]
\centering
\begin{tabular}{r|c|c}
& query & insertion \\
\hline
LSM-tree \cite{o'neil96:_lsm} with fractional cascading & &\\
Y-tree \cite{jermaine99:_novel} & $\log_\ell\frac{N}{M} + \frac{K}{B}$ & $
\frac{\ell}{B} 
\log_\ell \frac{N}{M}$ \\
buffer repository tree \cite{buschbaum:dfsdir} && \\
\hline
stepped merge tree \cite{jagadish97:_increm} with fractional cascading  &
$\ell \log_\ell\frac{N}{M} + \frac{K}{B}$ & $ \frac{1}{B} \log_\ell
\frac{N}{M}$ \\  \hline 
\end{tabular}
\caption{Query/insertion upper bounds of previously known B-tree
  indexes, for a parameter $2\le \ell\le B$. }
\label{tab:results}
\end{table}

In particular, the indexes listed in Table~\ref{tab:results} are all quite
practical, so one may wonder if there are some fancy complicated
theoretical structures with better bounds that have not been found yet.
For the static range query problem, it turned out to be indeed the case.  A
somehow surprising result by Alstrup et al.~\cite{alstrup:_optim} shows
that it is possible to achieve linear size and $O(K)$ query time in the RAM
model.  This results also carries over to 
external memory, yielding a disk-based index with $O(\frac{N}{B})$ blocks
and $O(1+ \frac{K}{B})$-I/O query cost.  However, this structure is overly
complicated, and is actually worse than the B-tree in practice.  In the
dynamic case, a recent result by Mortensen et al.~\cite{mortensen05:_dynam}
gives a RAM-structure with $O(\log\log\log N + K)$ query time and
$O(\log\log N)$ update time.  This result, when carried over to external
memory, gives us an update cost of $O(\log\log N)$ I/Os.  This could be
much worse than the $O(\frac{1}{B}\log\frac{N}{M})$ bound obtained by the
simple dynamic B-trees mentioned earlier, for typical values of
$N,M$, and $B$.  Until today no bounds better than the ones in
Table~\ref{tab:results} are known.  The  $O(\log\frac{N}{M}+\frac{K}{B})$
query and $O(\frac{1}{B}\log\frac{N}{M})$ 
insertion bounds seem to be an inherent barrier that has been standing
since 1996.  Nobody can break one without sacrificing the other.

Lower bounds for this and related problems have also been sought for.  For
lower bounds we will only consider insertions; the results will also hold
for the more general case where insertions and deletions are both present.
A closely related problem to range queries is the {\em predecessor}
problem, in which the index stores a set of keys, and the query asks
for the preceding key for a query point.  The predecessor problem has been
extensively studied in various internal memory models, and the bounds are
now tight in almost all cases \cite{beame02:_optim}.  
In external memory, Brodal and Fagerberg \cite{brodal03:_lower} prove that
for the dynamic predecessor problem, if insertions are handled in
$O(\frac{1}{B}\log\nm)$ I/Os amortized, a predecessor query has to take
$\Omega(\frac{\log(N/M)}{\log\log(N/M)})$ I/Os in the worst case.  Their
lower bound model is a comparison based external memory model.  However, a
closer look at their proof reveals that their techniques can actually be
adapted to prove the same lower bound of $\Omega(\frac{\log(N/M)}{\log\log
  (N/M)}+\frac{K}{B})$ for range queries for any $B=\omega(1)$.
More precisely, we can use their techniques to get the following tradeoff:
If an insertion takes amortized $u/B$ I/Os and a query takes worst-case
$q+O(\frac{K}{B})$ I/Os, then we have
\begin{equation}
\label{eq:previous}
q\cdot\log(u\log^2\tfrac{N}{M}) =\Omega(\log\tfrac{N}{M}),
\end{equation}
provided $u\le B/\log^3 N$ and $N\ge M^2$.  
In addition to \eqref{eq:previous}, a few other tradeoffs have also been
obtained in \cite{brodal03:_lower} for the predecessor problem, but their
proofs cannot be made to work for range queries.  For the most
interesting case when we require $q = O(\log\nm)$, \eqref{eq:previous}
gives a meaningless bound of $u=\Omega(1/\log^2\nm)$, as $u\ge 1$
trivially.  In the other direction, if $u =O(\log\nm)$, the tradeoff
\eqref{eq:previous} still leaves an $\Theta(\log\log \nm)$ gap to the known
upper bound for $q$.

\paragraph{Our results.}
In this paper, we prove a query-insertion tradeoff of 
\begin{equation}
\label{eq:1}
\left\{
\begin{array}{ll}
q \cdot \log(u/q) =\Omega(\log B), & \textrm{for } q < \alpha \ln B,
\textrm{where } \alpha \textrm{ is any constant}; \\
u \cdot \log q = \Omega(\log B), & \textrm{for all } q.
\end{array}\right.
\end{equation}
for any dynamic range query index with a query cost of $q +
O(K/B)$ and an amortized insertion cost of $u/B$, provided $N \ge 2MB^2$.
For most reasonable values of $N, M$, and $B$, we may assume that $\nm =
B^{O(1)}$, or equivalently that the B-tree built on $N$ keys has constant
height.  In this
case if we require $q=O(\log\nm) = O(\log B)$, the first branch of
\eqref{eq:1} gives $u=\Omega(\log B)$, matching the known upper bounds in
Table~\ref{tab:results}.  In the other direction, if $u=O(\log\nm) = O(\log
B)$, we have $q=\Omega(\log B) = \Omega(\log\nm)$, which is again tight,
and closes the $\Theta(\log\log \nm)$ gap left in \cite{brodal03:_lower}.
In fact for any $2\le \ell \le B$, if $u=O(\ell \log_\ell B)$, we have a
tight lower bound $q=\Omega(\log_\ell B)$, matching the bounds in the first
row of Table~\ref{tab:results}.  The second branch of \eqref{eq:1} is
relevant for larger values of $q$, for which the previous tradeoff
\eqref{eq:previous} is helpless.  In particular, if $u=O(\log_B \nm) =
O(1)$, we have $q = B^{\Omega(1)}$.  This means that if we want to support
very fast insertions, the query cost has to go from logarithmic to
polynomial, an exponential blowup.  This matches the second row of
Table~\ref{tab:results}.  Our results show that all the indexes listed in
Table~\ref{tab:results}, which are all quite simple and practical, are
already essentially the best one can hope for.



More interestingly, our lower bounds hold in a dynamic version of the {\em
  indexability model} \cite{hellerstein02:indexability}, which was
originally proposed by Hellerstein, Koutsoupias, and Papadimitriou
\cite{hellerstein:analysis}.  To date, nearly all the known lower bounds
  for indexing problems are proved in this model~\cite{hellerstein02:indexability,arge:indexability,arge:enclosure,samoladas:lower,koutsoupias:tight}.
It is in some sense the strongest possible model for reporting problems.
It basically assumes that the query cost is only determined by the number
of disk blocks that hold the actual query results, and ignores all the
search cost that we need to pay to find these blocks.  Consequently, lower
bounds obtained in this model are also stronger than those obtained in
other models.  We will give more details on this model in
Section~\ref{sec:dynamic-indexability}.  However, until today this model
has been used exclusively for studying static indexing problems and only in
two or higher dimensions.  In one dimension, the model yields trivial
bounds (see Section~\ref{sec:dynamic-indexability} for details).  In the
JACM article \cite{hellerstein02:indexability} that summarizes most of the
results on indexability, the authors state: ``However, our model also
ignores the dynamic aspect of the problem, that is, the cost of insertion
and deletion.  Its consideration {\em could} be a source of added
complexity, and in a more general model the source of more powerful lower
bounds.''  In this respect, another contribution of this paper is to add
dynamization to the model of indexability, making it more powerful and
complete.  In particular, our lower bound results suggest that, although
static indexability is only effective in two or more dimensions,
dynamization makes it a suitable model for one-dimensional indexing
problems as well.

\section{Dynamic Indexability}
\label{sec:dynamic-indexability}
\paragraph{Static indexability.}
We first briefly review the framework of indexability before introducing
its dynamization.  We follow the notations from
\cite{hellerstein02:indexability}.  A {\em workload} $W$ is a tuple
$W=(D,I,\Q)$ where $D$ is a possibly infinite set (the {\em domain}),
$I\subseteq D$ is a finite subset of $D$ (the {\em instance}), and $\Q$ is
a set of subsets of $I$ (the {\em query set}).  For example, for
one-dimensional range queries, $D$ is the real line, $I$ is a set of points
on the line, and $\Q$ consists of all the contiguous subsets of $I$.  We
usually use $N=|I|$ to denote the number of objects in the instance.  An
{\em indexing scheme} $\S=(W,\B)$ consists of a workload $W$ and a set $\B$
of $B$-subsets of $I$ such that $\B$ covers $I$.  The $B$-subsets of $\B$
model the data blocks of an index structure, while any auxiliary structures
connecting these data blocks (such as pointers, splitting elements) are
ignored from this framework.  The size of the indexing scheme is $|\B|$,
the number of blocks.  In \cite{hellerstein02:indexability}, an equivalent
parameter, the {\em redundancy} $r=B|\B|/N$ is used to measure the space
complexity of the indexing scheme.  The cost of a query $q\in \Q$ is the
minimum number of blocks whose union covers $q$.  Note that here we have
implicitly assumed that the query algorithm can find these blocks to cover
$q$ instantly with no cost, essentially ignoring the ``search cost''.  The
{\em access overhead} $A$ is the minimum $A$ such that any query $q\in\Q$
has a cost at most $A\cdot \lceil |q|/B \rceil$.  Note that $\lceil
|q|/B\rceil$ is the minimum number of blocks to report the objects in $q$,
so the access overhead $A$ measures how much more we need to access the
blocks in order to retrieve $q$.  For some problems using a single
parameter for the access overhead is not expressive enough, and we may
split it into two: one that depends on $|q|$ and another that does not.
More precisely, an indexing scheme with access overhead $(A_0, A_1)$ must
answer any query $q\in\Q$ with cost at most $A_0+A_1\cdot \lceil |q|/B
\rceil$ \cite{arge:enclosure}.  We can see that the indexability model is
very strong.  It is the strongest possible model that one can conceive for
reporting problems.  It is generally accepted that no index structure could
break indexability lower bounds, unless it somehow ``creates'' objects
without accessing the original ones or their copies.

Except for some trivial facts, all the lower bound results obtained under
this model are expressed as a tradeoff between $r$ and $A$ (or
$(A_0,A_1)$).  For example, two-dimensional range reporting has a tradeoff
of $r=\Omega(\log(N/B)/\log A)$
\cite{arge:indexability,hellerstein02:indexability}; for the {\em point
  enclosure} problem, the dual of range queries, we have the tradeoff
$A_0A_1^2 = \Omega(\log(N/B)/\log r)$ \cite{arge:enclosure}.  These results
show that, even if we ignore the search cost, we can obtain nontrivial
lower bounds for these problems.  These lower bounds have also been matched
with corresponding indexes that {\em do} include the search cost for
typical values of $r$ and $A$ \cite{arge:indexability,arge:enclosure}.
This means that the inherent difficulty for these indexing problems roots
from how we should {\em layout} the data objects on disk, not the search
structure on top of them.  By ignoring the search component of an index, we
obtain a simple and clean model, which is still powerful enough to reveal
the inherent complexity of indexing.  It should be commented that the
indexability model is very similar in spirit to the {\em cell probe model}
of Yao \cite{yao81:_shoul}, which has been successfully used to derive many
internal memory lower bounds.  But the two models are also different in
some fundamental ways; please see \cite{hellerstein02:indexability} for a
discussion.

Nevertheless, although the indexability model is appropriate for
two-dimensional problems, it seems to be overly strong for the more basic
one-dimensional range query problem.  In one dimension, we could simply
layout all the points in order sequentially on disk, which would give us a
linear-size, constant-query access overhead index!  This breaks the
$\Theta(\log_B N)$ bound of the good old B-tree, and suggests that the
indexability model may be too strong for studying one-dimensional
workloads.  This in fact can be explained.  The $\Omega(\log_B N)$ lower
bound holds only in some restrictive models, such as the comparison model,
and the B-tree indeed only uses comparisons to guide its search.  As we
mentioned in the introduction, if we are given more computational power
(such as direct addressing), we can actually solve the static 1D range
query problem with an index of linear size and $O(\lceil K/B
\rceil)$-I/O query cost \cite{alstrup:_optim}.  This means that the search
cost for 1D range queries can still be ignored without changing the
complexity of the problem, and the indexability model is still appropriate,
albeit it only gives a trivial lower bound.

\paragraph{Dynamic indexability.}
In the dynamic case, the domain $D$ remains static, but the instance set
$I$ could change.  Correspondingly, the query set $\Q$ changes and the
index also updates its blocks $\B$ to cope with the changes in $I$.  In the
static model, there is no component to model the main memory, which is all
right since the memory does not help reduce the worst-case query cost 
anyway.  However, in the dynamic case, the main memory does improve the
(amortized) update cost significantly by buffering the recent updates.  So
we have to include a main memory component in the indexing scheme.  More
precisely, the workload $W$ is defined as before, but an indexing scheme is
now defined as $\S=(W, \B, \M)$ where $\M$ is a subset of $I$ with size at
most $M$ such that the blocks of $\B$ together with $\M$ cover $I$.  The
redundancy $r$ is defined as before, but the access overhead $A$ is now
defined as the minimum $A$ such that any $q\in\Q$ can be covered by $\M$
and at most $A\cdot \lceil |q|/B \rceil$ blocks from $\B$.

We now define the {\em dynamic indexing scheme}.  Here we only consider
insertions; deletions can be incorporated similarly.  We first define the
{\em dynamic workload}.
\begin{definition}
  A {\em dynamic workload} $\mathbb{W}$ is a sequence of $N$ workloads
  $W_1=(D,I_1,\Q_1), \dots, W_N=(D,I_2,\Q_2)$ such that $|I_i| = i$ and
  $I_i \subset I_{i+1}$ for $i=1,\dots,N-1$.
\end{definition}
Essentially, we insert $N$ objects into $I$ one by one, resulting in a
sequence of workloads.  Meanwhile, the query set $\Q$ changes according
to the problem at hand.
\begin{definition}
  For a given dynamic workload $\mathbb{W} = (W_1,\dots, W_N)$, a {\em
    dynamic indexing scheme} $\mathbb{S}$ is a sequence of $N$ indexing
  schemes $\S_1 =(W_1, \B_1,\M_1), \dots, \S_N=(W_N,\B_N, \M_N)$.  Each
  $\S_i$ is called a {\em snapshot} of $\mathbb{S}$.  $\mathbb{S}$ has
  redundancy $r$ and access overhead $A$ if for all $i$, $\S_i$ has
  redundancy at most $r$ and access overhead at most $A$.
\end{definition}
A third parameter $u$, the {\em update cost}, is defined as follows.
\begin{definition}
  Given a dynamic indexing scheme $\mathbb{S}$, the {\em transition cost}
  from $\S_i$ to $\S_{i+1}$ is $|\B_i - \B_{i+1}| + |\B_{i+1} - \B_i|$,
  i.e., the number of blocks that are different in $\B_i$ and $\B_{i+1}$.
  The {\em update cost} $\mathbb{S}$ is the $u$ such that the sum of all
  the transition costs for all $1\le i\le N-1$ is $u\cdot N/B$.
\end{definition}
Note that the update cost as defined above is the amortized cost for
handling $B$ updates.  This is mainly for convenience so that
$u$ is always at least 1.

Our definition of the dynamic indexability model continues the same spirit
as in the static case: We will only focus on the cost associated with the
changes in the blocks holding the actual data objects, while ignoring the
search cost of how to find these blocks to be changed.  Under this
framework, the main result obtained in this paper is the following
tradeoff between $u$ and $A$.
\begin{theorem}
\label{thm:dynam-index}
Let $\mathbb{S}$ be any dynamic indexing scheme for dynamic one-dimensional
range queries with access overhead $A$ and update cost $u$.  Provided $N\ge
2MB^2$, we have
\[
\left\{
\begin{array}{ll}
A \cdot \log(u/A) =\Omega(\log B), & \textrm{for }  A < \alpha \ln B,
\textrm{where } \alpha \textrm{ is any constant}; \\
u \cdot \log A = \Omega(\log B), & \textrm{for all } A.
\end{array}\right.
\]
\end{theorem}
Note that this lower bound does not depend on the redundancy $r$, meaning
that the index cannot do better by consuming more space.  Interestingly,
our result shows that although the indexability model is basically
meaningless for static 1D range queries, it gives nontrivial and almost
tight lower bound when dynamization is considered.  

To prove Theorem~\ref{thm:dynam-index}, below we first define a {\em
  ball-shuffling} problem and show that any dynamic indexing scheme for
1D range queries yields a solution to the ball-shuffling problem.  Then we
prove a lower bound for the latter.

\section{The Ball-Shuffling Problem and the Reduction}
\label{sec:ball-shuffl-probl}
We now define the {\em ball-shuffling} problem, and present a lower bound
for it.  There are $n$ balls and $t$ bins, $b_1, \dots, b_t$.  The balls
come one by one.  Upon the arrival of each ball, we need to find some bin
$b_i$ to put it in.  Abusing notations, we use also $b_i$ to denote the
current size of the bin, i.e., the number of balls inside.  The {\em cost}
of putting the ball into $b_i$ is defined to be $b_i+1$.  Instead of
directly putting a ball into a bin, we can do so with {\em shuffling}: We
first collect all the balls from one or more bins, add the new ball to the
collection, and then arbitrarily allocate these balls into a number of
empty bins.  The cost of this operation is the total number of balls
involved, i.e., if $I$ denotes the set of indices of the bins collected,
the cost is $\sum_{i \in I}b_i + 1$.  Note that directly putting a ball
into a bin can be seen as a special shuffle, where we collect balls from
only one bin and allocate the balls back to one bin.

Our main result for the ball-shuffling problem is the following lower
bound, whose proof is deferred to Section~\ref{sec:proof}.

\begin{theorem}
\label{thm:shuffling}
The cost of any algorithm for the ball-shuffling problem is at least (i)
$\Omega(n\log_t n)$ for any $t$; and (ii) $\Omega(tn^{1+\Omega(1/t)})$ for
$t <\alpha\ln n$ where $\alpha$ is an arbitrary constant.
\end{theorem}

\paragraph{The reduction.}
Suppose there is a dynamic indexing scheme $\mathbb{S} = (\S_1,\dots,\S_N)$
for dynamic one-dimensional range queries with update cost $u$ and access
overhead $A$.  Assuming $N\ge 2MB^2$, we will show how this leads to a
solution to the ball-shuffling problem on $n=B$ balls and $t=A$ bins with
cost $O(uB)$.  This will immediately translate the tradeoff in
Theorem~\ref{thm:shuffling} to the desired tradeoff in
Theorem~\ref{thm:dynam-index}.

We divide these $N$ points into subsets of $2MB^2$.  We will use a separate
construction for each subset of points.  Since the amortized cost for
handling every $B$ insertions of points is $u$, at least one of the subsets
has a total transition cost of at most $O(u MB)$.  Let us consider one such
subset of $N'=2MB^2$ points.

We construct a dynamic workload of $N'$ points as follows.  The points are
divided into $2MB$ groups of $B$ each.  The coordinates of all points in
the $j$-th group are in the range of $(j,j+1)$ and distinct.  We perform
the insertions in $B$ rounds; in each round, we simply add one point to
each group.  The dynamic indexing scheme $\mathbb{S}$ correspondingly has
$N'$ snapshots $\S_1=(W_1, \B_1, \M_1),\dots,\S_{N'} =(W_{N'}, \B_{N'},
\M_{N'})$.  We will only consider the subsequence $\mathbb{S}'$ consisting
of the snapshots $\S_{2MB}, \S_{2\cdot 2MB},\dots, \S_{N'}$, i.e., the ones
after every round.  The total transition cost of this subsequence is
obviously no higher than that of the entire sequence.  Recall that the
transition cost from a snapshot $\S=(W,\B,\M)$ to its succeeding snapshot
$\S'=(W',\B',\M')$ is the number of blocks that are different in $\B$ and
$\B'$.  We now define the {\em element transition cost} to be the number of
elements in these different blocks, more precisely, $|\{x \mid x \in b, b
\in (\B-\B')\cup(\B'-\B)\}|$.  Since each block contains at most $B$
elements, the element transition cost is at most a factor $O(B)$ larger
than the transition cost.  Thus, $\mathbb{S}'$ has an element transition
cost of $O(uMB^2)$.  The element transition cost can be associated with the
elements involved, that is, it is the total number of times that an element
has been in an updated block, summed over all elements.

If a group $G$ has at least one point in some $\M_i$ in $\mathbb{S}'$, then
it is said to be {\em contaminated}.  Since $\sum_{i=1}^B |\M_{i\cdot 2MB}|
\le MB$, at most $MB$ groups are contaminated.  Since the total element
transition cost of $\mathbb{S}'$ is $O(uMB^2)$, among the at least $MB$
uncontaminated groups, at least one has an element transition cost of
$O(uB)$.  Focusing on such a group, and let $G_1,\dots,G_B$ be the
snapshots of this group after every round.  Since this group is
uncontaminated, all points in $G_i$ must be completely covered by
$\B_{i\cdot 2MB}$ for all $i=1,\dots, B$.  Since $G_i$ has at most $B$
points and $\mathbb{S}$ has access overhead $A$, $G_i$ should always be
covered by at most $A$ blocks in $\B_{i\cdot 2MB}$.  For each $i$, let
$b_{i,1}, \dots, b_{i,A}$ be the blocks of $\B_{i\cdot 2MB}$ that cover
$G_i$, let $\hat{b}_{i,j} = b_{i,j} \cap G_i, j=1,\dots,A$.  Note that
these $\hat{b}_{i,j}$ may overlap and some of them may be empty.  Let
$\hat{\B}_i = \{\hat{b}_{i,1}, \dots, \hat{b}_{i,A}\}$.  Consider the
transition from $\hat{\B}_i$ to $\hat{\B}_{i+1}$.  We can as before define
its element transition cost as $|\{x \mid x \in b, b \in
(\hat{\B}_i-\hat{\B}_{i+1})\cup(\hat{\B}_{i+1}-\hat{\B}_i)\}|$.  This
element transition cost cannot be higher than that from $\B_{i\cdot 2MB}$
to $\B_{(i+1)\cdot 2MB}$ only counting the elements of $G_{i+1}$, because
$\hat{b}_{i,j} \ne \hat{b}_{i,j'}$ only if $b_{i,j} \ne b_{i,j'}$.
Therefore, the total element transition cost of the sequence $\hat{\B}_1,
\dots, \hat{\B}_B$ is at most $O(uB)$.

Now we claim that the sequence $\hat{\B}_1, \dots, \hat{\B}_B$ gives us a
solution for the ball-shuffling problem of $B$ balls and $A$ bins with cost
at most its element transition cost.  To see this, just treat each
set in $\hat{\B}_i$ as a bin in the ball-shuffling problem.  To add the
$(i+1)$-th ball, we shuffle the bins in $\hat{\B}_i-\hat{\B}_{i+1}$ and
allocate the balls according to the sizes of the sets in $\hat{\B}_{i+1} -
\hat{\B}_i$.  An element may have copies in $\hat{\B}_{i+1}$, so there
could be more elements than balls in $\hat{\B}_{i+1} - \hat{\B}_i$.  But
this is all right, we can still allocate balls according to $\hat{\B}_{i+1}
- \hat{\B}_i$, while just making sure that each bin has no more balls than
their corresponding set in $\hat{\B}_{i+1}$.  This way, we can ensure that
the cost of each shuffle is always no more than the element transition cost
of each transition.  Therefore, we obtain a solution to the ball-shuffling
problem with cost $O(uB)$.  This completes the  reduction.

\section{Proof of Theorem~\ref{thm:shuffling}}
\label{sec:proof}
\paragraph{Proof of part (i).} We first prove part (i) of the theorem.  We
will take an indirect approach, proving that any algorithm that handles the
balls with an average cost of $u$ using $t$ bins cannot accommodate
$(2t)^{2u}$ balls or more.  This means that $n < (2t)^{2u}$, or $u >
\frac{\log n}{2\log(2t)}$, so the total cost of the algorithm is $un =
\Omega(n\log_t n)$.

We prove so by induction on $u$.  When $u=1$, clearly the algorithm has to
put every ball into an empty bin, so with $t$ bins, the algorithm can
handle at most $t < (2t)^2$ balls.  We will use a step size of
$\frac{1}{2}$ for the induction, i.e., we will assume that the claim is
true for $u$, and show that it is also true for $u+\frac{1}{2}$. (Thus our
proof works for any $u$ that is a multiple of $\frac{1}{2}$; for other
values of $u$, the lower bound becomes $(2t)^{\lceil 2u \rceil}$, which
does not affect our asymptotic result.)  Equivalently we need to show that
to handle $(2t)^{2u+1}$ balls, any algorithm using $t$ bins has to pay an
average cost of more than $u+\frac{1}{2}$ per ball, or
$(u+\frac{1}{2})(2t)^{2u+1} = (2tu+t)(2t)^{2u}$ in total.  We divide the
$(2t)^{2u+1}$ balls into $2t$ batches of $(2t)^{2u}$ each.  By the
induction hypothesis, to handle the first batch, the algorithm has to pay a
total cost of more than $u(2t)^{2u}$.  For each of the remaining batches,
the cost is also more than $u(2t)^{2u}$, plus the cost of shuffling the
existing balls from previous batches.  This amounts to a total cost of
$2tu(2t)^{2u}$, and we only need to show that shuffling the balls from
previous batches costs at least $t(2t)^{2u}$ in total.

If a batch has at least one ball that is never shuffled in later batches,
it is said to be a {\em bad batch}, otherwise it is a {\em good batch}.
The claim is that at most $t$ of these $2t$ batches are bad.  Indeed, since
each bad batch has at least one ball that is never shuffled later, the bin
that this ball resides in cannot be touched any more.  So each bad batch
takes away at least one bin from later batches and there are only $t$ bins.
Therefore there are at least $t$ good batches, in each of which all the
$(2t)^{2u}$ ball have been shuffled later.  This costs at least
$t(2t)^{2u}$, and the proof completes.

\paragraph{The merging lemma.}
Part (i) of the theorem is very loose for small values of $t$.  If $t\le
\alpha \log n$ where $\alpha$ is an arbitrary constant, we can prove a much
higher lower bound, which later will lead to the most interesting branch in
the query-update tradeoff \eqref{eq:1} of range queries.  The rest of this
section is devoted to the proof of part (ii) of
Theorem~\ref{thm:shuffling}, and it requires a much more careful and direct
analysis.

We first prove the following lemma, which restricts the way how the optimal
algorithm might do shuffling.  We call a shuffle that allocates balls back
to more than one bin a {\em splitting shuffle}, otherwise it is a {\em
  merging} shuffle.

\begin{lemma}
\label{lem:merge}
There is an optimal algorithm that only uses merging shuffles.
\end{lemma}
\begin{proof}
  For a shuffle, we call the number of bins that receive balls from the
  shuffle its {\em splitting number}.  A splitting shuffle has a splitting
  number at least 2, and a merging shuffle's splitting number is 1.  For an
  algorithm $\A$, let $\pi(\A)$ be the sequence of the splitting numbers of
  all the $n$ shuffles performed by $\A$.  Below we will show how to
  transform $\A$ into another algorithm $\A'$ whose cost is no higher than
  that of $\A$, while $\pi(\A')$ is lexicographically smaller than
  $\pi(\A)$.  Since every splitting number is between 1 and $t$, after a
  finite number of such transformations, we will arrive at an algorithm
  whose splitting numbers are all 1, hence proving the lemma.
  
  Let $\A$ be an algorithm that uses at least one splitting shuffle, and
  consider the last splitting shuffle carried out by $\A$.  Suppose it
  allocates  balls to $k$ bins.  $\A'$ will do the same as $\A$ up until
  its last splitting shuffle, which $\A'$ will change to the following
  shuffle.  $\A'$ will collect balls from the same bins but will only
  allocate them to $k-1$ bins.  Among the $k-1$ bins, $k-2$ of them receive
  the same number of balls as in $\A$, while the last bin receives all the
  balls in the last two bins used in $\A$.  Observe that since the bins are
  indistinguishable, the current status of the bins is only determined by
  their sizes.  So the only difference between $\A$ and $\A'$ after this
  shuffle is two bins, say $b_1, b_2$ of $\A$ and $b'_1, b'_2$ of $\A'$.
  Note that the cost of this shuffle is the same for both $\A$ and $\A'$.
  After this shuffle, suppose we have $b_1= x, b_2=y, b'_1 = x+y, b'_2 = 0$
  for some $x, y \ge 1$.  Clearly, no matter what $\A'$ does in the future,
  we always have $\pi(\A')$ lexicographically smaller than $\pi(\A)$.
  
  From now on $\A'$ will mimic what $\A$ does with no higher cost.  We will
  look ahead at the operations that $\A$ does with $b_1$ and $b_2$, and
  decide the corresponding actions of $\A'$.  Note that $\A$ will do no
  more splitting shuffles.  Consider all the shuffles that $\A$ does until
  it merges $b_1$ and $b_2$ together, or until the end if $\A$ never does
  so.  For those shuffles that touch neither $b_1$ nor $b_2$, $\A'$ will
  simply do the same.  Each of the rest of the shuffles involves $b_1$ but
  not $b_2$ (resp.\ $b_2$ but not $b_1$).  Since the bins are
  indistinguishable, for any such merging shuffle, we may assume that all
  the balls are put back to $b_1$ (resp.\ $b_2$).  Suppose there are $a_1$
  shuffles involving $b_1$ and $a_2$ shuffles involving $b_2$.  Assume for
  now that $a_1 \le a_2$.  $\A'$ will do the following correspondingly.
  When $\A$ touches $b_1$, $\A'$ will use $b'_1$; and when $\A$ touches
  $b_2$, $\A'$ will use $b'_2$.  Clearly, for any shuffle that involves
  neither $b_1$ nor $b_2$, the cost is the same for $\A$ and $\A'$.  For a
  shuffle that involves $b_1$ but not $b_2$, since before $\A$ merges $b_1$
  and $b_2$, we have the invariant that $b'_1 = b_1 + y$, $\A'$ pays a cost
  of $y$ more than that of $\A$, for each of these $a_1$ shuffles.  For a
  shuffle that involves $b_2$ but not $b_1$, since we have the invariant
  that $b'_2 = b_2 - y$, $\A'$ pays a cost of $y$ less than that of $\A$,
  for each of these $a_2$ shuffles.  So $\A'$ incurs a total cost no more
  than that of $\A$.  In the case $a_1 \ge a_2$, when $\A$ touches $b_1$,
  $\A'$ will use $b'_2$; and when $\A$ touches $b_2$, $\A'$ will use
  $b'_1$.  A similar argument then goes through.  Finally, when $\A$ merges
  $b_1$ and $b_2$ together (if it ever does so), $\A'$ will also shuffle
  both $b'_1$ and $b'_2$.  Since we always have $b_1 + b_2 = b'_1+b'_2$,
  the cost of this shuffle is the same for $\A$ and $\A'$.  After this
  shuffle, $\A$ and $A'$ are in the same status.  Thus we have transformed
  $\A$ into $\A'$ with no higher cost while $\pi(\A')$ is strictly
  lexicographically smaller than $\pi(A)$.  Applying such transformations
  iteratively proves the lemma.
\end{proof}

\paragraph{The recurrence.}
Now we are ready to prove part (ii) of Theorem~\ref{thm:shuffling}.  Our
general approach is by induction on $t$.  Let $f_t(n)$ be the minimum cost
of any algorithm for the ball-shuffling problem with $n$ balls and $t$
bins.  Let $\alpha$ be an arbitrary constant.  The induction process
consists of two phases.  In the first phase, we prove that $f_t(n) \ge c_1
tn^{1+c_2/t} - 2tn$ for all $t \le t_0 =\lfloor c_0 \ln n \rfloor$, where
$c_0, c_1$ and $c_2$ are some small constants to be determined later.  In
phase two, we prove that $f_t(n) \ge c_1 t_0 n^{1+c_2/(t_0+(t-t_0)/\alpha)}
- 2tn$ for all $t_0 \le t\le \alpha \ln n$.  Finally we show how to choose
the constants $c_0, c_1, c_2$ such that $f_t(n)$ is always at least
$\Omega(t n^{1+\Omega(1/t)})$.

The base case of the first phase $t=1$ is easily established, since the
optimal algorithm is simply adding the balls to the only bin one by one,
yielding $f_1(n) = \frac{1}{2}n(n+1) \ge c_1 n^{1+c_2} - 2n$, provided that
we choose $c_1 \le 1/2, c_2 \le 1$.

By Lemma~\ref{lem:merge}, there is an optimal algorithm $\A$ for shuffling
$n$ balls with $t+1$ bins where $\A$ only uses merging shuffles.  Since the
bins are indistinguishable, we may assume w.l.o.g.\ that there is a
designated bin, say $b_1$, such that whenever $b_1$ is shuffled, all the
balls are put back to $b_1$.  Suppose when handling the last ball, we force
$\A$ to shuffle all the balls to $b_1$, which costs $n$.  We will later
subtract this cost since $\A$ may not actually do so in the last step.

Suppose $\A$ carries out a total of $k$ shuffles involving $b_1$ (including
the last enforced shuffle), and with the $i$-th shuffle, $b_1$ increases by
$x_i\ge 1$.  It is clear that $\sum_{i=1}^k x_i= n$.  We claim that the
total cost of $\A$, $f_{t+1}(n)$, is at least
\begin{equation}
\label{eq:2}
f_t(x_1) + f_t(x_2) + \cdots + f_t(x_k) + \left(k-\frac{1}{t}\right)x_1 +
\left(k-1-\frac{1}{t}\right)x_2 + \cdots + \left(1-\frac{1}{t}\right)x_{k}
- 2n.
\end{equation}

Consider the $i$-th shuffle involving $b_1$.  This shuffle brings $x_i$
balls to $b_1$, including the new ball just added in this step.  Let us
lower bound the cost due to these $x_i$ balls.  First, those $x_i-1$ old
balls must not have been in $b_1$ before, since whenever $\A$ shuffles
$b_1$, all the balls will go back to $b_1$.  So $\A$ must have been able to
accommodate them using the other $t$ bins.  This costs at least
$f_t(x_i-1)$, even if ignoring the cost of shuffling the other existing
balls in these $t$ bins.  Then these $x_i-1$ balls, plus a new ball, are
shuffled to $b_1$.  This costs $x_i$, not counting the cost associated with
the existing balls in $b_1$.  Finally, these $x_i$ balls will be in $b_1$
for all of the remaining $k-i$ shuffles involving $b_1$, costing $(k-i)
x_i$.  Thus, we can charge a total cost of
\begin{equation}
\label{eq:3}
f_t(x_i - 1) + x_i +
(k-i)x_i = f_t(x_i-1) + 1 + \frac{x_i}{t} + \left(k-i+1-\frac{1}{t}\right)x_i-1 \ge f_t(x_i) + \left(k-i+1 - \frac{1}{t}\right) x_i -1
\end{equation}
to these $x_i$ balls.  That $f_t(x_i-1) + 1 + x_i/t \ge f_t(x_i)$ easily
follows from the observation that, to handle $x_i$ balls with $t$ bins, we
can always run the optimal algorithm for $x_i-1$ balls with $t$ bins, and
then put the last ball into the smallest bin, which will cost no more than
$1 + (x_i-1)/t < 1+x_i/t$.  Finally, summing (\ref{eq:3}) over for all $i$,
relaxing a $-k$ to $-n$, and subtracting the cost of the enforced shuffle
proves that (\ref{eq:2}) is a lower bound on $f_{t+1}(n)$ for given $k,
x_1, \dots, x_k$.  Thus, $f_{t+1}(n)$ is lower bounded by the minimum of
\eqref{eq:2}, over all possible values of $k, x_1,\dots,x_k$, subject to
$\sum_{i=1}^k x_i = n$.

We first use this recurrence to solve for $f_2(n)$.  
\begin{eqnarray*}
 f_2(n) &\ge& \min_{k, x_1+\cdots+x_k = n} \{ f_1(x_1) + \cdots + f_1(x_k)
 + (k-1)x_1 + \cdots + x_{k-1} - 2n\} \\
&=& \min_{k,x_1+\cdots+x_k =n} \{\frac{1}{2}x_1(x_1+1) + \cdots +
 \frac{1}{2}x_k(x_k+1) + (k-1)x_1 + \cdots + x_{k-1} -2n\} \\
& \ge & \min_{k}\left\{\frac{1}{2} k \left(\frac{n}{k}\right)^2 +
 \frac{1}{2}(k-1)k -2n \right \} \ge \frac{1}{4}n^{4/3} - 2n.
\end{eqnarray*}
So if we choose $c_1\le 1/4$, $c_2 \le 2/3$, we have $f_t(n) \ge c_1 t
n^{1+c_2/t} -2tn$ for $t=2$.

For $t\ge 2$, we relax the recurrence as
\begin{eqnarray}
 f_{t+1}(n) &\ge& \min_{k,x_1+\cdots+x_k = n}\left\{ f_t(x_1)+ \cdots +
  f_t(x_k) + \left(k-\frac{1}{2}\right) x_1 +
  \left(k-1-\frac{1}{2}\right)x_2 + \cdots +\frac{1}{2} x_{k}-2n\right\}
\nonumber \\
\label{eq:4}
&\ge& \min_{k,x_1+\cdots+x_k = n}\{ f_t(x_1)+ \cdots +
  f_t(x_k) + \frac{1}{2}(k x_1 + (k-1)x_2+ \cdots + x_k) - 2n\}. 
\end{eqnarray} 

\paragraph{The induction, phase one.}
In phase one, we have $1\le t \le t_0-1$ for $t_0 = \lfloor c_0 \ln
n\rfloor$.  The base cases $t=1,2$ have already been established.  Assuming
the induction hypothesis $f_t(n) \ge c_1 t n^{1+c_2/t} - 2tn$, we need to
show $f_{t+1}(n) \ge c_1 (t+1) n^{1+c_2/(t+1)} -2(t+1)n$.  From
(\ref{eq:4}) we have
\begin{equation}
\label{eq:5}
f_{t+1}(n) \ge \min_{k, x_1+\cdots +x_k = n}\{ c_1 t x_1^{1+c_2/t} -2tx_1 +
\cdots + c_1 t x_k^{1+c_2/t} - 2tx_k + \frac{1}{2}(k x_1 + \cdots +
x_k)-2n\}. 
\end{equation}
Let $g_{k}(n)$ be the minimum of (\ref{eq:5}) for a given $k$.  Then
clearly $f_{t+1}(n) \ge \min_{1\le k \le n} g_{k}(n)$, and we will show
that 
\begin{equation}
\label{eq:6}
g_k(n) \ge c_1 (t+1)  n^{1+c_2/(t+1)} - 2(t+1)n
\end{equation}
for all $k$, hence completing the induction.

We prove so using another level of induction on $k$.  For the base case
$k=1$, we have $g_1(n) \ge c_1 t n^{1+c_2/t} - 2tn + \frac{1}{2} n -2n \ge
c_1t n^{1+c_2/t} -2(t+1)n$, and $c_1 t n^{1+c_2/t} \ge c_1 (t+1)
n^{1+c_2/(t+1)}$ holds as long as
\[ tn^{\frac{c_2}{t}} \ge (t+1) n^{\frac{c_2}{t+1}} 
\quad \Leftrightarrow \quad
n^{\frac{c_2}{t(t+1)}} \ge 1+\frac{1}{t}
\quad \Leftrightarrow \quad
n^{\frac{c_2}{t+1}} \ge \left(1+\frac{1}{t}\right)^t 
\quad \Leftarrow \quad
n^{\frac{c_2}{t+1}} > e
\quad \Leftrightarrow \quad
t \le c_2 \ln n - 1.
\]
So if we choose $c_0 < c_2$, then for the range of $t$ that we consider in
phase one, (\ref{eq:6}) holds for $k=1$.

Next, assuming that (\ref{eq:6}) holds for $k$, we will show $g_{k+1}(n)
\ge c_1(t+1)n^{1+c_2/(t+1)} - 2(t+1)n$.  By definition,
\begin{eqnarray*}
 g_{k+1}(n) &=& \min_{x_1+\cdots + x_{k+1} = n} \{ c_1 tx_1^{1+c_2/t} -
 2tx_1+  \cdots  + c_1 tx_{k+1}^{1+c_2/t} - 2tx_{k+1} + \frac{1}{2}( (k+1)
 x_1+ \cdots  +  x_{k+1}) -2n\}\\
 &=& \min_{x_{k+1}} \{ c_1t x_{k+1}^{1+c_2/t} - 2tx_{k+1} + \frac{1}{2}n +
 \min_{x_1+\cdots  + x_k = n-x_{k+1}} \{ c_1 tx_1^{1+c_2/t}-2t x_1 + \cdots
 +  c_1  tx_k^{1+c_2/t} -2tx_k\\
&& + \frac{1}{2}( k x_1+ \cdots + x_k) - 2 (n-x_{k+1})\} - 2x_{k+1}\}\\
&=& \min_{x_{k+1}} \{ c_1t x_{k+1}^{1+c_2/t}  -2(t+1)x_{k+1} + \frac{1}{2}n
 + g_k(n-x_{k+1})\}  \\ 
&\ge& \min_{x_{k+1}} \{ c_1t x_{k+1}^{1+c_2/t} -2(t+1)x_{k+1}+ \frac{1}{2}n
 +  c_1(t+1)(n-x_{k+1})^{1+c_2/(t+1)}-2(t+1)(n-x_{k+1})\} \\
&=& \min_{x_{k+1}} \{ c_1t x_{k+1}^{1+c_2/t} + \frac{1}{2}n
 +  c_1(t+1)(n-x_{k+1})^{1+c_2/(t+1)}-2(t+1)n\}. \\
\end{eqnarray*}
Setting $x_{k+1} = \lambda n$ where $0 < \lambda < 1$, we will show
\begin{equation}
\label{eq:7} c_1 t (\lambda n)^{1+c_2/t} +
c_1(t+1)((1-\lambda) n)^{1+c_2/(t+1)} + \frac{1}{2}n
\ge c_1(t+1) n^{1+c_2/(t+1)}
\end{equation}
for all $\lambda$.  (\ref{eq:7}) is equivalent to 
\begin{equation}
\label{eq:8}
 \frac{t}{t+1} \lambda^{1+\frac{c_2}{t}} n^{\frac{c_2}{t(t+1)}} +
(1-\lambda)^{1+\frac{c_2}{t+1}}  + \frac{1}{2c_1(t+1)n^{c_2/(t+1)}}
\ge 1.
\end{equation}
Since $(1-\lambda)^{1+\frac{c_2}{t+1}} \ge (1-\lambda)^{1+\frac{c_2}{t}}$,
to prove (\ref{eq:8}), it suffices to prove
\begin{equation}
\label{eq:9}
 \frac{t}{t+1} n^{\frac{c_2}{t(t+1)}}  \lambda^{1+\frac{c_2}{t}}  +
(1-\lambda)^{1+\frac{c_2}{t}} \ge 1 -
\frac{1}{2c_1(t+1)n^{c_2/(t+1)}}.
\end{equation}

The LHS of (\ref{eq:9}) achieves its only minimum at the point where its
derivative is zero, namely when
\begin{eqnarray}
\nonumber
 \frac{t}{t+1} n^{\frac{c_2}{t(t+1)}} \left(1+\frac{c_2}{t}\right)
 \lambda^{\frac{c_2}{t}} 
 &=& \left(1+\frac{c_2}{t}\right)(1-\lambda)^{\frac{c_2}{t}},\\
\nonumber
\textrm{or} \qquad \left(\frac{t}{t+1}\right)^{t/c_2} n^{1/(t+1)} \lambda
 & =& 1-\lambda,\\ 
\label{eq:10}
\lambda &=& \frac{1}{(\frac{t}{t+1})^{t/c_2} n^{1/(t+1)} + 1}.
\end{eqnarray}
Plugging (\ref{eq:10}) into the LHS of (\ref{eq:9}) while letting 
$\gamma = (\frac{t}{t+1})^{t/c_2} n^{1/(t+1)}$, we get
\[
\frac{\gamma^{c_2/t} +  \gamma^{1+c_2/t}}
{(\gamma+1)^{1+c_2/t}}  = 
\frac{\gamma^{c_2/t}(1+\gamma)} {(\gamma+1)^{1+c_2/t}} = 
\frac{\gamma^{c_2/t}} {(\gamma+1)^{c_2/t}} =
\left(\frac{\gamma}{\gamma+1}\right)^{c_2/t}.
\]
Considering the RHS of (\ref{eq:9}), since $n^{c_2/(t+1)} =
\gamma^{c_2}(\frac{t+1}{t})^t < e \gamma^{c_2}$, we have
\[ 1 - \frac{1}{2c_1(t+1)n^{c_2/(t+1)}} = 1-\frac{1}{2c_1 (t+1)
  \gamma^{c_2}(\frac{t+1}{t})^t} < 1 - \frac{1}{2c_1 e (t+1)\gamma^{c_2}}
<  1 - \frac{1}{4c_1 e t\gamma^{c_2}}.
\]
Thus, to have (\ref{eq:9}), we just need to have 
\begin{eqnarray*}
  \left(\frac{\gamma}{\gamma+1}\right)^{c_2/t}& \ge& 1 - \frac{1}{4c_1 e
  t\gamma^{c_2}}, \\
\textrm{or} \qquad \frac{\gamma}{\gamma+1} &\ge& \left(1 - \frac{1}{4c_1 e
  t\gamma^{c_2}} \right)^{t/c_2} =
\left(1 - \frac{1}{4c_1 e t\gamma^{c_2}}
  \right)^{\frac{4c_1et\gamma^{c_2}}{4c_1 c_2 e\gamma^{c_2}}}  \\
\Leftarrow \qquad \frac{\gamma}{\gamma+1} &\ge & \exp\left(-\frac{1}{4c_1
  c_2 e \gamma^{c_2}}\right) \\
\Leftrightarrow \qquad 
 1+\frac{1}{\gamma} &\le & \exp\left(\frac{1}{4c_1 c_2
  e\gamma^{c_2}}\right) \\ 
\Leftarrow \qquad 
 1+\frac{1}{\gamma} &\le & 1 + \frac{1}{4c_1 c_2 e\gamma^{c_2}},
\end{eqnarray*}
where the last inequality holds if $\gamma \ge 4c_1 c_2 e\gamma^{c_2}$, or
$\gamma \ge (4c_1c_2e)^{1/(1-c_2)}$.  Finally, since 
\[ \gamma =  n^{1/(t+1)} \left/ \left(1+1/t\right)^{t/c_2}\right. >
n^{1/(t+1)} / e^{1/c_2} \ge n^{1/t_0}/e^{1/c_2} \ge e^{1/c_0 - 1/c_2},
\]
as long as we choose $c_0$ small enough depending on $c_1$ and $c_2$, such
that $ e^{1/c_0 - 1/c_2} \ge (4c_1c_2e)^{1/(1-c_2)}$, \eqref{eq:9} will
hold, and henceforth $g_{k+1}(n) \ge c_1(t+1)n^{1+c_2/(t+1)}$.  This
also completes the induction on $t$ for phase one.  Finally, to ensure $c_1
tn^{1+c_2/t} - 2tn =\Omega(tn^{1+\Omega(1/t)})$ for $t\le t_0$, it suffices
to have $c_1 n^{c_2/t_0} = c_1 e^{c_2/c_0}> 2$, which again can be
guaranteed by choosing $c_0$ small enough.

\paragraph{The induction, phase two.}
The derivation for phase two is similar to that of phase one, and is
given in the appendix.  Combining the results of phase one and phase two
we have proved part (ii) of Theorem~\ref{thm:shuffling}.

\paragraph{Tightness of the bounds.}
Ignoring the constants in the Big-Omega, the lower bound of
Theorem~\ref{thm:shuffling} is tight for nearly all values of $t$.  Now we
give some concrete strategies matching the lower bounds For $t
\ge 2\log n$, we use the following shuffling strategy.  Let $x = t/\log n
\ge 2$.  Divide the $t$ bins evenly into $\log_x n$ groups of $t/\log_x n$
each.  We use the first group to accommodate the first $t/\log_x n$ balls.
Then we shuffle these balls to one bin in the second group.  In general,
when all the bins in group $i$ are occupied, we shuffle all the balls in
group $i$ to one bin in group $i+1$.  The total cost of this algorithm is
obviously $n \log_x n$ since each ball has been to $\log_x n$ bins, one
from each group.  To show that this algorithm actually works, we need to
show that all the $n$ balls can be indeed accommodated.  Since the capacity
of each group increases by a factor of $t/\log_x n$, the capacity of the
last group is
\[ \left(\frac{t}{\log_x n}\right)^{\log_x n} =  \left(\frac{x t}{x \log_x
    n}\right)^{\log_x n} = n \left(\frac{t}{x \log_x n}\right)^{\log_x n} =
n \left(\frac{\log n}{\log_x n}\right)^{\log_x n} = n (\log x)^{\log_x n}
\ge n.
\]
Thus, part (i) of Theorem~\ref{thm:shuffling} is tight as long as
$\log(t/\log n) = \Omega(\log t)$, or $t=\Omega(\log^{1+\eps}n)$.

Part (ii) of the theorem concerns with $t=O(\log n)$.  For such a small $t$
we need to deploy a different strategy.  We always put balls one by one to
the first bin $b_1$.  When $b_1$ has collected $n^{1/t}$ balls, we shuffle
all the balls to $b_2$.  Afterward, every time $b_1$ reaches $n^{1/t}$, we
merge all the balls in $b_1$ and $b_2$ and put the balls back to $b_2$.
For $b_2$, every time it has collected $n^{2/t}$ balls from $b_1$, we merge
all the balls with $b_3$.  In general, every time $b_i$ has collected
$n^{i/t}$ balls, we move all the balls to $b_{i+1}$.  Let us compute the
total cost of this strategy.  For each shuffle, we charge its cost to the
destination bin.  Thus, the cost charged to $b_1$ is at most $(n^{1/t})^2
\cdot n^{1-1/t} = n^{1+1/t}$, since for every group of $n^{1/t}$ balls, it
pays a cost of at most $(n^{1/t})^2$ to add them one by one, and there are
$n^{1-1/t}$ such groups.  In general, for any bin $b_i, 1\le i \le t$, the
balls arrive in batches of $n^{(i-1)/t}$, the bin clears itself for every
$n^{1/t}$ such batches.  The cost for each batch is at most $n^{i/t}$, the
maximum size of $b_i$, so the cost for all the $n^{1/t}$ batches before
$b_i$ clears itself is $n^{(i+1)/t}$.  The bin clears itself $n/n^{i/t} =
n^{1-i/t}$ times, so the total cost charged to $b_i$ is $n^{1+1/t}$.
Therefore, the total cost charged to all the bins is $tn^{1+1/t}$.

Combining part (i) and part (ii), our lower bound is thus tight for all $t$
except in the narrow range $\omega(\log n) \le t \le o(\log^{1+\eps}n)$.
And in this range, the gap between the upper and lower bounds is merely
$\Theta(\frac{\log t}{\log(t/\log n)}) = o(\log\log n)$.



\bibliographystyle{abbrv} 
\bibliography{../paper,../io,../geom,../newgeom}

\appendix
\section{The induction, phase two}
In phase two, we will prove that $f_t(n) \ge c_1 t_0
n^{1+c_2/(t_0+c_0(t-t_0)/\alpha)} - 2tn$ for $t_0\le t \le \alpha \ln n$
where $\alpha$ is any given constant.  To simplify notations we define
$h(t) = t_0 + c_0(t-t_0)/\alpha$.  The base case $t=t_0$ for phase two has
already been established from phase one.  Next we assume $f_t(n) \ge c_1
t_0 n^{1+c_2/h(t)} - 2tn$ and will prove that $f_{t+1}(n) \ge c_1 t_0
n^{1+c_2/h(t+1)}-2(t+1)n$.

From the recurrence \eqref{eq:4} and the induction hypothesis, we have
\begin{equation}
\label{eq:11}
f_{t+1}(n) \ge \min_{k, x_1+\cdots +x_k = n}\{ c_1 t_0
x_1^{1+c_2/h(t)} - 2tx_1 +
\cdots + c_1 t_0 x_k^{1+c_2/h(t)} - 2tx_k + \frac{1}{2}(k
x_1 + \cdots + x_k) - 2n\}. 
\end{equation}
Similarly as in phase one, let $g_k(n)$ be the minimum of~\eqref{eq:11} for
a given $k$.  Here we need to show that
\begin{equation}
\label{eq:12}
g_k(n) \ge c_1 t_0  n^{1+c_2/h(t+1)} - 2(t+1)n.
\end{equation}

Again we use induction on $k$ to prove \eqref{eq:12}.  The base case is
easily seen as $g_1(n) = c_1 t_0 n^{1+c_2/h(t)} -2tn +
\frac{1}{2}n -2n > c_1 t_0 n^{1+c_2/h(t+1)}-2(t+1)n$.  Now
suppose \eqref{eq:12} holds for $k$, we will show $g_{k+1}(n) \ge c_1 t_0
n^{1+c_2/h(t+1)} - 2(t+1)n$.  By the induction hypothesis,
we have
\begin{eqnarray*}
 g_{k+1}(n) &=& \min_{x_1+\cdots + x_{k+1} = n} \{ c_1 t_0x_1^{1+c_2/h(t)}
 - 2tx_1 + \cdots  + c_1 t_0x_{k+1}^{1+c_2/h(t)} -
 2tx_{k+1} \\
&& + \frac{1}{2}( (k+1) x_1+ \cdots +  x_{k+1}) - 2n\}\\
 &=& \min_{x_{k+1}} \{ c_1  t_0x_{k+1}^{1+c_2/h(t)} -
 2tx_{k+1} + \frac{1}{2}n + \min_{x_1+\cdots+x_k = n-x_{k+1}}\{c_1 
 t_0x_1^{1+c_2/h(t)} - 2tx_1 +   \\ 
&& \cdots  +c_1 t_0x_k^{1+c_2/h(t)} - 2tx_k 
 + \frac{1}{2}( k x_1+ \cdots +  x_k) - 2(n-x_{k+1})\} - 2x_{k+1}\}\\
&=& \min_{x_{k+1}} \{ c_1t_0 x_{k+1}^{1+c_2/h(t)} -2(t+1)
 x_{k+1}+ 
 \frac{1}{2}n + g_k(n-x_{k+1}) \}  \\  
&\ge& \min_{x_{k+1}} \{ c_1t_0 x_{k+1}^{1+c_2/h(t)} -2(t+1)x_{k+1} +
 \frac{1}{2}n  +  c_1 t_0(n-x_{k+1})^{1+c_2/h(t+1)} -  2(t+1)(n-x_{k+1})\}
 \\  
&=&  \min_{x_{k+1}} \{ c_1t_0 x_{k+1}^{1+c_2/h(t)}  +
 \frac{1}{2}n  +  c_1 t_0(n-x_{k+1})^{1+c_2/h(t+1)} -
 2(t+1)n\}.
\end{eqnarray*}
Setting $x_{k+1} = \lambda n$ where $0 < \lambda < 1$, we will show
\begin{equation}
\label{eq:13} c_1 t_0 (\lambda n)^{1+c_2/h(t)} +
c_1 t_0((1-\lambda) n)^{1+c_2/h(t+1)} + \frac{1}{2}n
\ge c_1 t_0 n^{1+c_2/h(t+1)}
\end{equation}
for all $\lambda$.  (\ref{eq:13}) is equivalent to 
\begin{equation}
\label{eq:14}
 \lambda^{1+c_2/h(t)} n^{\frac{c_2c_0/\alpha}{h(t)h(t+1)}} +
(1-\lambda)^{1+c_2/h(t+1)}  + \frac{1}{2c_1 t_0n^{c_2/h(t+1)}}
\ge 1.
\end{equation}
Since $(1-\lambda)^{1+\frac{c_2}{h(t+1)}} \ge
(1-\lambda)^{1+\frac{c_2}{h(t)}}$, 
to prove (\ref{eq:14}), it suffices to prove
\begin{equation}
\label{eq:15}
  n^{\frac{c_2c_0/\alpha}{h(t)h(t+1)}}  \lambda^{1+\frac{c_2}{h(t)}}  +
(1-\lambda)^{1+\frac{c_2}{h(t)}} \ge 1 -
\frac{1}{2c_1 t_0 n^{c_2/h(t+1)}}.
\end{equation}

The LHS of (\ref{eq:15}) achieves its only minimum when
\begin{eqnarray}
\nonumber
 n^{\frac{c_2c_0/\alpha}{h(t)h(t+1)}} \left(1+\frac{c_2}{h(t)}\right)
 \lambda^{\frac{c_2}{h(t)}} 
 &=& \left(1+\frac{c_2}{h(t)}\right)(1-\lambda)^{\frac{c_2}{h(t)}},\\
\nonumber
\textrm{or} \qquad  n^{\frac{c_0/\alpha}{h(t+1)}} \lambda
 & =& 1-\lambda,\\ 
\label{eq:16}
\lambda &=& \frac{1}{n^{\frac{c_0/\alpha}{h(t+1)}} + 1}.
\end{eqnarray}
Plugging (\ref{eq:16}) into (\ref{eq:15}) while letting $\gamma =
n^{\frac{c_0/\alpha}{h(t+1)}}$, \eqref{eq:15} becomes
\begin{eqnarray*}
  \left(\frac{\gamma}{\gamma+1}\right)^{c_2/h(t)}& \ge&  1-\frac{1}{2c_1
    t_0 \gamma^{c_2\alpha/c_0}}, \\ 
\textrm{or} \qquad \frac{\gamma}{\gamma+1} &\ge& \left(1 - \frac{1}{2c_1
    t_0\gamma^{c_2\alpha/c_0}} \right)^{h(t)/c_2} = 
\left(1 - \frac{1}{2c_1 t_0\gamma^{c_2\alpha/c_0}}
  \right)^{\frac{2c_1t_0\gamma^{c_2\alpha/c_0}h(t)}{2c_1 t_0
    \gamma^{c_2\alpha/c_0}c_2}}  \\ 
\Leftarrow \qquad  \frac{\gamma}{\gamma+1} &\ge&
    \exp\left(-\frac{h(t)}{2c_1c_2  t_0 \gamma^{c_2\alpha/c_0}}\right) \\
\Leftarrow \qquad  \frac{\gamma}{\gamma+1} &\ge&
    \exp\left(-\frac{1}{2c_1c_2 \gamma^{c_2\alpha/c_0}}\right)
    \\ 
\Leftarrow \qquad  1+ \frac{1}{\gamma} &\le&
    1+ \frac{1}{2c_1c_2 \gamma^{c_2\alpha/c_0}}, 
\end{eqnarray*}
where the last inequality holds if $\gamma \ge 2c_1 c_2
\gamma^{c_2\alpha/c_0}$.  We will choose $c_2,c_0$ such that $c_2\alpha/c_0
> 1$, thus this becomes $\gamma \le
(\frac{1}{2c_1c_2})^{\frac{1}{c_2\alpha/c_0-1}}$.  Since $\gamma =
n^{\frac{c_0/\alpha}{h(t+1)}} < n^{\frac{c_0/\alpha}{c_0\ln n}} =
e^{1/\alpha}$, we just need to have $e^{1/\alpha} \le
(\frac{1}{2c_1c_2})^{\frac{1}{c_2\alpha/c_0-1}}$ to make sure that
\eqref{eq:15} holds.  This would also complete the induction on $t$ for
phase two.

We also need to ensure that $c_1t_0 n^{1+c_2/h(t)} - 2tn \ge c_1 c_0/\alpha
\cdot t n^{1+c_2/h(t)}= \Omega(t n^{1+\Omega(1/t)})$ for phase two.  This
just requires $c_1c_0/\alpha \cdot n^{c_2/h(t)} > 2$.  Since $c_1c_0/\alpha
\cdot n^{c_2/h(t)} \ge c_1c_0/\alpha \cdot
n^{\frac{c_2}{(2c_0-c_0^2/\alpha)\ln n}} = c_1c_0/\alpha \cdot
e^{\frac{c_2}{2c_0-c_0^2/\alpha}}$, we just require $c_1c_0/\alpha \cdot
e^{\frac{c_2}{2c_0-c_0^2/\alpha}} > 2$.

\smallskip
Finally, we put together all the constraints that we have on the constants:
\[\left\{\begin{array}{l}
c_1 \le 1/2, c_2 \le 1, \\
c_1 \le 1/4, c_2 \le 2/3, \\
c_0 < c_2, \\
(4c_1c_2e)^{1/(1-c_2)} \le e^{1/c_0 - 1/c_2}, \\
2 < c_1 e^{c_2/c_0}, \\
e^{1/\alpha} \le (\frac{1}{2c_1c_2})^{\frac{1}{c_2\alpha/c_0-1}}, \\
2 < c_1c_0/\alpha \cdot e^{\frac{c_2}{2c_0-c_0^2/\alpha}}. \\
\end{array}\right.\]
We can first fix $c_1=c_2= 1/4$.  This makes $(4c_1c_2e)^{1/(1-c_2)} < 1$.
Then we choose $c_0$ small enough such that the third and the fifth
constraints are satisfied.  That $c_0 < c_2$ also makes $e^{1/c_0 - 1/c_2}
\ge 1$, satisfying the fourth constraint.  Finally, we will make $c_0$ even
smaller if necessary (depending on $\alpha$), to satisfy the last two
constraints.  This completes the proof of part (ii) of
Theorem~\ref{thm:shuffling}.

\end{document}